\documentclass[envcountsame]{llncs}

\usepackage{epsfig}
\usepackage{latexsym}
\usepackage{amsmath}
\usepackage{amsfonts}
\usepackage{amssymb}
\usepackage{graphicx}

\newcommand{\mc}{\mathcal}
\newcommand{\Oh}{{\mc{O}}}

\newcommand{\Fh}{{\mc{F}}}

\newtheorem{fact}{Fact}
\newtheorem{obs}{Observation}

\title{On the Parameterized Complexity of Reconfiguration Problems}
\begin{document}

\author
{
    Amer E. Mouawad\inst{1}\thanks{Research supported by the Natural Science and Engineering Research Council of Canada.} \and
    Naomi Nishimura\inst{1}$^{\star}$ \and
    Venkatesh Raman\inst{2} \and
    Narges Simjour\inst{1}$^{\star}$ \and
    Akira Suzuki\inst{3}\thanks{Research supported by JSPS Grant-in-Aid for Scientific Research, Grant Number 24.3660.}
}
\institute
{
    David R. Cheriton School of Computer Science\\
    University of Waterloo, Waterloo, Ontario, Canada.\\
    \email{\{aabdomou, nishi, nsimjour\}@uwaterloo.ca}
    \and
    The Institute of Mathematical Sciences\\
    Chennai, India.
    \email{vraman@imsc.res.in}
    \and
    Graduate School of Information Sciences, Tohoku University\\
    Aoba-yama 6-6-05, Aoba-ku, Sendai, 980-8579, Japan. \\
    \email{a.suzuki@ecei.tohoku.ac.jp}
}

\maketitle
\sloppy

\begin{abstract}
  We present the first results on the parameterized complexity of
  reconfiguration problems, where a reconfiguration version of an
  optimization problem $Q$ takes as input two feasible solutions $S$
  and $T$ and determines if there is a sequence of {\em
    reconfiguration steps} that can be applied to transform $S$ into
  $T$ such that each step results in a feasible solution to $Q$.
 For
  most of the results in this paper, $S$ and $T$ are subsets of
  vertices of a given graph and a reconfiguration step adds or deletes
  a vertex.  Our study is motivated by recent results establishing
  that for most NP-hard problems, the classical complexity of
  reconfiguration is PSPACE-complete.

  We address the question for several important graph properties under
  two natural parameterizations: $k$, the size of the solutions, and
  $\ell$, the length of the sequence of steps.  Our first general
  result is an algorithmic paradigm, the {\em reconfiguration kernel},
  used to obtain fixed-parameter algorithms for the reconfiguration
  versions of {\sc Vertex Cover} and, more generally, {\sc Bounded
    Hitting Set} and {\sc
    Feedback Vertex Set}, all parameterized by $k$. In contrast, we
  show that reconfiguring {\sc Unbounded Hitting Set} is $W[2]$-hard
  when parameterized by $k+\ell$. We also demonstrate the
  $W[1]$-hardness of the reconfiguration versions of a large class of
  maximization problems parameterized by $k+\ell$, and of their
  corresponding deletion problems parameterized by $\ell$; in doing
  so, we show that there exist problems in FPT when parameterized by
  $k$, but whose reconfiguration versions are $W[1]$-hard when
  parameterized by $k+\ell$.
\end{abstract}

\section{Introduction}
The reconfiguration version of an optimization problem asks whether it is
possible to transform a source feasible solution $S$ into a target
feasible solution $T$ by a sequence of {\em reconfiguration steps} such that
every intermediate solution is also feasible; other variants return
a (possibly minimum-length) {\em reconfiguration sequence} of solutions.
Reconfiguration problems model real-life dynamic situations in which
we seek to transform a solution into a more desirable one, maintaining
feasibility during the process.  The study of reconfiguration yields
insights into the structure of the solution space of the underlying
optimization problem, crucial for the design
of efficient algorithms.

Motivated by real world situations as well as by trying to
understand the structure of all feasible solutions, there has been a
lot of recent interest in studying the complexity of reconfiguration
problems.  Problems for which reconfiguration has been studied include
{\sc Vertex Colouring}~\cite{BB13,BC09,CVJ08,CVJ09,CVJ11}, {\sc List
  Edge-Colouring}~\cite{IKD12}, {\sc Independent
  Set}~\cite{HD05,IDHPSUU11}, {\sc Clique}, {\sc Set Cover}, {\sc
  Matching}, {\sc Matroid Bases}~\cite{IDHPSUU11}, {\sc
  Satisfiability}~\cite{GKMP09}, {\sc Shortest Path}~\cite{B12,KMM11},
and {\sc Dominating Set}~\cite{HS12,SMN13}.  Most work has been
limited to the problem of determining the existence of a
reconfiguration sequence between two given solutions;
for most NP-complete problems, this problem has been shown to be
PSPACE-complete.

As there are typically exponentially many feasible solutions,
the length of the reconfiguration sequence can be exponential in
the size of the input instance.  It is thus natural to ask whether
reconfiguration problems become tractable if we allow the running time
to depend on the length of the sequence; this approach suggests the use of the
paradigm of parameterized complexity.
In this work, we explore reconfiguration in the framework of
parameterized complexity~\cite{DF97}
under two natural parameterizations:
$k$, a
bound on the size of feasible solutions, and $\ell$, the length
of the reconfiguration sequence.
One of our key results is that for most problems, the reconfiguration
versions remain intractable in the parameterized framework when we
parameterize by $\ell$.  It is important to note that when $k$ is not
bounded, the reconfiguration we study become easy.

We present fixed-parameter algorithms for problems parameterized by
$k$ by modifying known parameterized algorithms for the problems.  The
paradigms of bounded search tree and kernelization typically work by
exploring minimal solutions. However, a reconfiguration sequence may
necessarily include non-minimal solutions.  Any kernel that removes
solutions (non-minimal or otherwise) may render finding a
reconfiguration sequence impossible, as the missing solutions might
appear in every reconfiguration sequence; we must thus ensure that the
kernelization rules applied retain enough information to allow us to
determine whether a reconfiguration sequence exists.  To handle these
difficulties, we introduce a general approach for parameterized
reconfiguration problems.
We use a {\em reconfiguration kernel},
showing how to adapt Bodlaender's cubic kernel~\cite{B07} for {\sc
  Feedback Vertex Set}, and a special kernel by Damaschke and
Molokov~\cite{D09} for {\sc Bounded Hitting Set} (where the cardinality of
 each input set is bounded) to obtain polynomial reconfiguration kernels,
with respect to $k$.
These results can
be considered as interesting applications of kernelization,
 and a general approach for other similar reconfiguration problems.


As a counterpart to our result for {\sc Bounded Hitting Set}, we show that
reconfiguring {\sc Unbounded Hitting Set} or {\sc Dominating Set}
is $W[2]$-hard parameterized by $k + \ell$ (Section~\ref{sec-relate}).
Finally, we show a general result on reconfiguration problems of hereditary
properties and their `parametric duals', implying the $W[1]$-hardness of
reconfiguring {\sc Independent Set}, {\sc Induced Forest}, and {\sc Bipartite Subgraph}
parameterized by $k + \ell$ and {\sc Vertex Cover}, {\sc Feedback Vertex Set},
and {\sc Odd Cycle Transversal} parameterized by $\ell$.

\section{Preliminaries}\label{sec-prelims}

Unless otherwise stated, we assume that each input graph $G$ is a
simple, undirected graph on $n$ vertices with vertex set $V(G)$ and
edge set $E(G)$. To avoid confusion, we refer to {\em nodes} in
reconfiguration graphs (defined below), as distinguished from {\em
  vertices} in the input graph. We use the modified big-Oh notation
$O^*$ that suppresses all polynomially bounded factors.

Our definitions are based on optimization
problems, each consisting of a
polynomial-time recognizable set of valid instances, a set of feasible
solutions for each instance, and an objective function assigning a
nonnegative rational value to each feasible solution.

\begin{definition}
  The {\em reconfiguration graph}  $R_Q(I,\textnormal{adj},k)$,
  consists of a node for each feasible
  solution to instance $I$ of optimization problem $Q$, where the size
  of each solution is at least $k$ for $Q$ a maximization problem
  (of size at most $k$ for $Q$ a minimization problem,
  respectively), for positive integer $k$,
  and an edge between each pair of nodes corresponding to solutions in
  the binary adjacency relation $\textnormal{adj}$ on feasible solutions.
\end{definition}
We define the following {\em reconfiguration problems}, where $S$
and $T$ are feasible solutions for $I$: {\sc $Q$ Reconfiguration}
determines if there is a path from $S$ to $T$ in
$R_Q(I,\textnormal{adj},k)$; the {\em search variant} returns a
{\em reconfiguration sequence}, the sequence of feasible solutions
associated with such a path; and the {\em shortest path variant}
returns the reconfiguration sequence associated with a path of
minimum length.  For convenience, solutions paired by
$\textnormal{adj}$ are said to be {\em adjacent}.

Using the framework developed by Downey and Fellows~\cite{DF97}, a
{\em parameterized reconfiguration problem} includes in the input a
positive integer $\ell$ (an upper bound on the length of the
reconfiguration sequence) and a parameter $p$ (typically $k$ or
$\ell$). For a parameterized problem $Q$ with inputs of the form $(x,
p)$, $|x| = n$ and $p$ a positive integer, $Q$ is {\em
  fixed-parameter tractable} (or in {\em FPT}) if it can be decided in
$f(p) n^c$ time, where $f$ is an arbitrary function and $c$ is a
constant independent of both $n$ and $p$.  $Q$ is in the class {\em
  XP} if it can be decided in $n^{f(p)}$ time.  $Q$ has a
{\em kernel} of size $f(p)$ if there is an algorithm $A$ that
transforms the input $(x, p)$ to $(x', p')$ such that $A$ runs in
polynomial time (with respect to $|x|$ and $p$) and $(x, p)$ is a
yes-instance if and only if $(x', p')$ is a yes-instance, $p' \leq
g(p)$, and $|x'| \leq f(p)$. Each problem in {\em FPT} has a kernel,
possibly of exponential (or worse) size.

We introduce the related notion of a {\em reconfiguration kernel}; it
follows from the definition that a reconfiguration problem that has
such a kernel is in {\em FPT}.

\begin{definition}
  A {\em reconfiguration kernel} of an instance of a parameterized
  reconfiguration problem $(x,p) = (P,\textnormal{adj},S,T,k,\ell,p)$ is a set
  of $h(p)$ instances, for an arbitrary function $h$, such that for $1 \le i \le h(p)$:
\begin{itemize}
\item{} for each instance in the set,
  $(x_i,p_i) = (P,\textnormal{adj},S_i,T_i,k_i,\ell_i,p_i)$, the
  values of  $S_i$, $T_i$, $k_i$,
  $\ell_i$, and $p_i$ can all be computed in polynomial
  time,
\item{} the size of each $x_i$ is bounded by $j(p)$, for an arbitrary
  function $j$, and
\item{} $(x,p)$ is a yes-instance if and only if at least one $(x_i,p_i)$ is a yes-instance.
\end{itemize}
\end{definition}

The main hierarchy of parameterized
complexity classes is $FPT \subseteq W[1] \subseteq W[2] \subseteq
\ldots \subseteq XP$, where $W$-hardness, shown using {\em FPT
  reductions}, is the analogue of NP-hardness in classical complexity.
A parameterized
problem $Q$ {\em FPT reduces} to a parameterized problem $Q'$ if there
is an algorithm $A$ that transforms an instance $(I,p)$ of $Q$ to an
instance $(I', p')$ of $Q'$ such that $A$ runs in time $f(p)
poly(|I|)$ where $f$ is a function of $k$, $poly$ is a polynomial
function, and $p' = g(p)$ for some function $g$.  In addition, the
transformation has the property that $(I,p)$ is a yes-instance of $Q$
if and only if $(I', p')$ is a yes-instance of $Q'$. It is known that
standard parameterized versions (are there $p$ vertices that form the
solution?) of {\sc Clique} and {\sc Independent Set} are complete for
the class $W[1]$, and {\sc Dominating Set} is $W[2]$-complete.  The
reader is referred to~\cite{FG06,N06} for more on parameterized
complexity.


Most problems we consider can be defined
using graph properties,
where a {\em graph property} $\pi$ is a collection of graphs, and is
{\em non-trivial} if it is non-empty and does not contain all graphs.
A graph property is {\em polynomially decidable} if
for any graph $G$, it can be decided in polynomial time whether $G$ is in $\pi$.
For a subset $V' \subseteq V$, $G[V']$ is the
{\em subgraph of $G$ induced on $V'$}, with vertex set $V'$ and edge
set $\{\{u,v\} \in E \mid u,v \in V'\}$.  The property $\pi$ is {\em hereditary} if for any
$G \in \pi$, any induced subgraph of $G$ is also in $\pi$.  Examples
of hereditary properties include graphs having
no edges and graphs having no cycles.
It is well-known~\cite{LY80} that every hereditary property $\pi$
has a forbidden set $\Fh_\pi$, in that a graph has property
$\pi$ if and only if it does not contain any graph in $\Fh_\pi$ as an induced
subgraph.

For a graph property $\pi$, we define two reconfiguration graphs,
where solutions are sets of vertices and two solutions are
adjacent if they differ by the addition or deletion of a
vertex. The {\em subset reconfiguration
graph of $G$ with respect to $\pi$},
$R^{\pi}_{\textsc{sub}}(G,k)$, has a node for each
$S \subseteq V(G)$ such that $|S| \ge k$ and $G[S]$ has
property $\pi$, and the {\em deletion reconfiguration graph of $G$ with
respect to $\pi$}, $R^{\pi}_{\textsc{del}}(G,k)$,
has a node for each $S \subseteq
V(G)$ such that $|S| \le k$ and $G[V(G) \setminus S]$ has property
$\pi$.
We can obtain  $R^{\pi}_{\textsc{del}}(G,|V(G)|-k)$
by replacing the set corresponding to each node in
$R^{\pi}_{\textsc{sub}}(G,k)$ by its (setwise) complement.
The following
is a consequence of the fact that two nodes can differ by the
deletion or addition of a single vertex.

\begin{fact}\label{fact-degree-bound}
The degree of each node in $R^{\pi}_{\textsc{sub}}(G,k)$ and each node in
$R^{\pi}_{\textsc{del}}(G,k)$ is at most $|V(G)|$.
\end{fact}

\begin{definition}
For any graph property $\pi$, graph $G$, positive integer
$k$, $S \subseteq V(G)$, and $T \subseteq V(G)$, we define the following
decision problems:

\noindent \textsc{$\pi$-deletion$(G, k)$}: Is there $V' \subseteq V(G)$ such that
$|V'| \leq k$ and $G[V(G) \setminus V'] \in \pi$?\\
\smallskip
\noindent \textsc{$\pi$-subset$(G, k)$}: Is there $V' \subseteq V(G)$ such that
    $|V'| \geq k$ and $G[V'] \in \pi$? \\
\smallskip
\noindent \textsc{$\pi$-del-reconf$(G, S, T, k, \ell)$}: For $S,T \in
V(R^{\pi}_{\textsc{del}}(G,k))$,
is there a path of length at most $\ell$ between the nodes
  for $S$ and $T$ in $R^{\pi}_{\textsc{del}}(G,k)$?\\
\smallskip
\noindent \textsc{$\pi$-sub-reconf$(G, S, T, k, \ell)$}: For $S, T \in
V(R^{\pi}_{\textsc{sub}}(G,k))$,
is there a path of length at most
  $\ell$ between the nodes for $S$ and $T$ in $R^{\pi}_{\textsc{sub}}(G,k)$?
\end{definition}

We say that \textsc{$\pi$-deletion$(G, k)$} and \textsc{$\pi$-subset$(G, k)$}
are {\em parametric duals} of each other.
Note that in \textsc{$\pi$-subset$(G, k)$}, we seek a set of
vertices of size {\it at least} $k$ inducing a subgraph in $\pi$,
whereas in \textsc{$\pi$-deletion$(G, k)$}, we seek a
set of vertices of size {\it at most} $k$  whose {\it complement set} induces a
subgraph in $\pi$.
We refer to \textsc{$\pi$-del-reconf$(G, S, T, k, \ell)$} and
\textsc{$\pi$-sub-reconf$(G, S, T, k, \ell)$}
as {\em reconfiguration problems for $\pi$};
for example, for $\pi$ the set of graphs with no edges,
the former is \textsc{Vertex Cover Reconfiguration} and
the latter is \textsc{Independent Set Reconfiguration}.

\section{Fixed-Parameter Tractability Results}\label{sec-kernel}
We first observe that for any polynomially decidable graph property,
the \textsc{$\pi$-deletion} and \textsc{$\pi$-subset}
reconfiguration versions are in $XP$ when parameterized by $\ell$;
we conduct breadth-first search on the reconfiguration graph
starting at $S$, stopping either upon discovery of $T$ or upon
completing the exploration of $\ell$ levels.
Fact~\ref{fact-degree-bound} implies a bound of at most $n^{\ell}$ vertices to explore in total.

\begin{fact}
\label{xpclaim}
For any polynomially decidable graph property $\pi$,
\textsc{$\pi$-del-reconf$(G, S, T, k, \ell)$} $\in XP$
and
\textsc{$\pi$-sub-reconf$(G, S, T, k, \ell)$} $\in XP$
when parameterized by $\ell$.
\end{fact}

For an instance $(G,S,T,k,\ell)$, we partition $V(G)$ into the sets
$C = S \cap T$ (vertices common to $S$ and $T$), $S_D = S \setminus C$
(vertices to be deleted from $S$ in the course of reconfiguration),
$T_A = T \setminus C$ (vertices to be added to form $T$), and $O =
V(G) \setminus (S \cup T) = V(G) \setminus (C \cup S_D \cup T_A)$ (all
other vertices). Furthermore, we can partition $C$ into two sets $C_F$
and $C_M = C \setminus C_F$, where a vertex is in $C_F$ if and only if
it is in every feasible solution of size bounded by $k$.

The following fact is a consequence of the definitions above, the fact
that $\pi$ is hereditary, and the observations that $G[S_D]$ and $G[O]$
are both subgraphs of $G[V(G) \setminus T]$, and $G[T_A]$ and $G[O]$
are both subgraphs of $G[V(G) \setminus S]$.

\begin{fact}\label{fact-pieces}
For an instance \textsc{$\pi$-del-reconf$(G, S, T, k, \ell)$} of a
reconfiguration problem for hereditary property $\pi$, $G[O]$,
$G[S_D]$, and $G[T_A]$ all have property $\pi$.
\end{fact}

In any reconfiguration sequence, each
vertex in $S_D$ must be deleted and each vertex in $T_A$ must be
added.  We say that a reconfiguration sequence
{\em touches} a vertex $v$ if $v$ is either added or deleted in at least
one reconfiguration step.  Any vertex that is not touched is
{\em untouched}. In fact, since $\ell$ implies a bound on the total number
of vertices that can be touched in a reconfiguration sequence,
setting $\ell = |S_D| + |T_A|$ drastically simplifies the problem.

\begin{obs}\label{obs-touch-once}
For any polynomially decidable hereditary graph property $\pi$, if $|S_D| + |T_A| = \ell$, then
\textsc{$\pi$-del-reconf$(G, S, T, k, \ell)$} and \textsc{$\pi$-sub-reconf$(G, S, T, k, \ell)$}
can be solved in $\Oh^*(2^\ell)$ time, and hence are in FPT when parameterized by $\ell$.
\end{obs}

\begin{proof}
Since each vertex in $T_A$ must be added and each vertex in $S_D$ deleted,
in $\ell$ steps we can touch each vertex in $S_D \cup T_A$ exactly once; all
vertices in $V(G) \setminus (S_D \cup T_A)$ remain untouched.

Any node in the path between $S$ and $T$ in $R^{\pi}_{\textsc{sub}}(G,k)$
represents a set $C \cup B$ where $B$ is a subset of $S_D \cup
T_A$. As $|S_D| + |T_A| = \ell$, there are only $2^\ell$ choices for $B$.
Our problem then reduces to finding the shortest path between $S$ and
$T$ in the subgraph of $R^{\pi}_{\textsc{sub}}(G,k)$ induced on the $2^\ell$ relevant
nodes; the bound follows from the fact that the number of edges is at
most $2^{\ell}|V(G)|$, a consequence of Fact~\ref{fact-degree-bound}.
The same argument holds for $R^{\pi}_{\textsc{del}}(G,k)$.\qed
\end{proof}
In contrast, we show in the next section that for most hereditary properties,  reconfiguration
problems are hard when parameterized by $\ell$.

\subsection{Bounded Hitting Set}

 Here, we prove the parameterized tractability of reconfiguration for
 certain superset-closed $k$-subset problems when parameterized by
 $k$, where a {\em $k$-subset problem} is a parameterized problem $\sc
 Q$ whose solutions for an instance $(I,k)$ are all subsets of size at
 most $k$ of a domain set, and is {\em superset-closed} if any
 superset of size at most $k$ of a solution of $\sc Q$ is also a
 solution of $\sc Q$. For example, parameterized \textsc{Vertex Cover}
 is a superset-closed problem.


\begin{theorem}
\label{fullkernelresult}
If a $k$-subset problem $\sc Q$ is superset-closed and has an FPT algorithm to
enumerate all its minimal solutions, the number of which is bounded by a function of
$k$, then {\sc $Q$ Reconfiguration} parameterized by $k$ is in FPT,
as well as the search and shortest path variants.
\end{theorem}

\begin{proof}
  By enumerating all minimal
  solutions of $\sc Q$,
  we compute the set $M$ of all elements $v$ of the domain set such
  that $v$ is in a minimal solution to $\sc Q$.
  For $(I,S,T,k,\ell)$ an instance of
  $\sc Q$ \textsc{Reconfiguration},  we show that there exists a reconfiguration
  sequence from $S$ to $T$ if and only if there exists a reconfiguration
  sequence from $S\cap M$ to $T\cap M$ that uses only
  subsets of $M$.

Each set $U$ in the reconfiguration sequence
  from $S$ to $T$ is a solution, 
   hence
  contains at least one minimal solution in $U\cap M$;
 $U \cap M$ is a superset of the minimal solution and hence also a solution.
  Moreover, since
  any two consecutive solutions $U$ and $U'$ in the sequence differ by a single
  element, $U \cap M$ and $U' \cap M$ differ by at most a single element.
  By replacing each subsequence of identical sets by a single set,
  we obtain a reconfiguration sequence from $S \cap M$ to
  $T \cap M$ that uses only subsets of $M$.

  The reconfiguration sequence from $S\cap M$ to $T\cap M$ using only
  subsets of $M$ can be extended to a reconfiguration sequence from
  $S$ to $T$ by transforming $S$ to $S\cap M$ in $|S \setminus M|$
  steps and transforming $T\cap M$ to $T$ in $|T \setminus M|$
  steps. 
  In this
  sequence, each vertex in $C \setminus M$ is removed from $S$ to form
  $S \setminus M$ and then readded to form $T$ from $T \setminus M$.
  For each vertex $v \in C \setminus M$, we can choose instead to add $v$
  to each solution in the sequence, thereby decreasing $\ell$ by two
  (the steps needed to remove and then readd $v$) at the cost of
  increasing by one the capacity used in the sequence from $S \cap M$
  to $T \cap M$.
  This choice can be made independently for each
  of these ${\cal E} = |C\setminus M|$ vertices.

  Consequently, $(I,S,T,k,\ell)$ is a yes-instance for $\sc Q$
  \textsc{Reconfiguration} if and only if one of the ${\cal E} +1$
  reduced instances $(I, S\cap M, T\cap M, k-e, \ell-2({\cal E} -
  e))$, for $0 \le e \le {\cal E}$ and ${\cal E} = |C \backslash M|$, is a
  yes-instance for $\sc Q'$ \textsc{Reconfiguration}: we define $\sc Q'$
  as a $k$-subset problem whose solutions for an instance $(I,k)$
   are solutions of instance $(I,k)$ of $\sc Q$ that are contained in $M$.
  To show that $\sc Q'$
  \textsc{Reconfiguration} is in {\em FPT}, we observe that the number of nodes in the reconfiguration graph for $\sc Q'$ is bounded by a function of $k$:
  each solution of $\sc Q'$
  is a subset of $M$, yielding at most $2^{|M|}$ nodes,
  and $|M|$ is bounded by a function of $k$. \qed
\end{proof}

As a consequence, {\sc Bounded Hitting Set Reconfiguration}, {\sc Feedback Arc Set Reconfiguration in Tournaments Reconfiguration}, and  {\sc Minimum Weight SAT in Bounded CNF Formulas Reconfiguation} (where each solution is the set of variables that are set to true in a satisfying assignment, and the problem looks for a solution of cardinality at most $k$) are proved to be in FPT when parameterized by $k$:

\begin{corollary}
{\sc Bounded Hitting Set Reconfiguration},
{\sc Feedback Arc Set in Tournaments Reconfiguration}, and {\sc Minimum Weight SAT in
  Bounded CNF Formulas Reconfiguation} parameterized by $k$ are in FPT.
\end{corollary}
\begin{proof}
All these problems are superset-closed. Furthermore, standard techniques give FPT algorithms to enumerate their minimal solutions, 
 and the number of minimal solutions is bounded by a function of $k$ in all cases, as required by Theorem~\ref{fullkernelresult}.
We include the proofs for completeness.

We can devise a search tree algorithm that gradually constructs
minimal hitting sets of instances of {\sc Bounded Hitting Set},
producing all minimal hitting sets of size at most $k$ in its
leaves. Consider an instance of {\sc Bounded Hitting Set}, where the
cardinality of each set is bounded by a constant $c$. At each non-leaf
node, the algorithm chooses a set that is not hit, and branches on all
possible ways of hitting this set, including one of the (at most $c$)
elements in the set in each branch. Since we are not interested in
hitting sets of cardinality more than $k$, we do not need to search
beyond depth $k$ in the tree, proving an upper bound of $c^k$ on the number
of leaves, and an upper bound of $O^*(c^k)$ on the enumeration time.

For the problem {\sc Feedback Arc Set in Tournaments}, a tournament is
acyclic if and only if it has a directed cycle of length
three~\cite{Bang2008}, and a set of arcs is a minimal feedback arc set
in a tournament if and only if reversing its arcs in the tournament
results in an acyclic tournament~\cite{Raman2006}. Therefore, at each
non-leaf node in a search tree for this problem, there is always a
cycle $C$ of length three and every feedback arc set shares at least one
arc with $C$. The algorithm can thus branch on the three arcs
in $C$, reversing one in each branch, and solve the problem
recursively. As in the previous algorithm,
since we are not interested in feedback arc sets
of cardinality more than $k$, the search can be terminated at depth
$k$, proving an upper bound of $3^k$ on the number of minimal
$k$-feedback arc sets in tournaments, and an upper bound of $O^*(3^k)$
on the running time of this enumeration algorithm.

Finally, Misra et al.~\cite{Misra2013} give a search tree algorithm for
bounded CNF formula instances of {\sc Minimum Weight SAT}, where every
clause has at most $c$ literals for some constant $c$. At each node, the
algorithm chooses a clause whose literals are all positive, and
branches on all possible ways of satisfying the clause, setting one
variable to true in each branch. If there is no such clause, the
formula is satisfied with no increase in the number of true variables,
by setting every non-assigned variable to false. As before,
the algorithm stops the search when it reaches a depth of $k$,
proving an upper bound of $c^k$ on the number of satisfying
assignments, and an upper bound of $O^*(c^k)$ on the enumeration
time. \qed
\end{proof}

For {\sc Bounded Hitting Set}, the proof of
Theorem~\ref{fullkernelresult} can be strengthened to develop a
polynomial reconfiguration
kernel. In fact, we use the ideas in Theorem~\ref{fullkernelresult} to
adapt a special kernel that retains all minimal $k$-hitting sets in the reduced
instances~\cite{D09}.


\begin{theorem}
\label{hittingsettheorem}
{\sc Bounded Hitting Set Reconfiguration} parameterized by $k$ has a
polynomial reconfiguration kernel.
\end{theorem}

\begin{proof}
  We let $(G,S,T,k,\ell)$ be an instance of {\sc Bounded Hitting Set
    Reconfiguration}: $G$ is a family of sets of vertices of size at
  most $r$ and each of $S$ and $T$ is a hitting set of size at most
  $k$, that is, a set of vertices intersecting each set in $G$.  We
  form a reconfiguration kernel using the reduction algorithm $\cal A$
  of Damaschke and Molokov~\cite{D09}: $G' = {\cal A}(G)$ contains all
  minimal hitting set solutions of size at most $k$, and is of size at
  most $(r-1)k^r + k$.

  {\sc Bounded Hitting Set} is a $k$-subset problem that is
  superset-closed. Moreover, $V(G')$ includes all minimal
  $k$-hitting sets, and the $k$-hitting sets for $G'$ are actually
  those $k$-hitting sets for $G$ that are completely included
  in $V(G')$.
  Therefore, as in the proof of Theorem~\ref{fullkernelresult},
  $(G,S,T,k,\ell)$ is a yes-instance for
  {\sc Bounded Hitting Set Reconfiguration} if and only if
  one of the ${\cal E} +1$ reduced instances $(G', S\cap V(G'),
  T\cap V(G'), k-e, \ell-2({\cal E} - e))$, for $0 \le e \le {\cal E}$,
  is a yes-instance for {\sc Bounded Hitting Set Reconfiguration}.

  Notice that unlike in the proof of Theorem~\ref{fullkernelresult},
  here we have access to an $f(k)$-bounded instance $G'$ based on
  which we can solve {\sc $Q'$ Reconfiguration}.  Another difference
  is that here the set containing all minimal solutions can be
  computed in polynomial time, whereas Theorem~\ref{fullkernelresult}
  guarantees only a fixed-parameter tractable procedure.  \qed
\end{proof}

{\sc Bounded Hitting Set} generalizes {\sc Vertex Cover}, {\sc
  Feedback Vertex Set in Tournaments}, {\sc Cluster Deletion}, and in
general any deletion problem for a hereditary property with a finite
forbidden set:
\begin{corollary}\label{finiteforbidencor}
If $\pi$ is a hereditary graph property with a finite forbidden set,
then \textsc{$\pi$-del-reconf$(G, S, T, k, \ell)$} parameterized by $k$ has a
polynomial reconfiguration kernel.
\end{corollary}

\subsection{Undirected Feedback Vertex Set}

Corollary~\ref{finiteforbidencor} does not apply to {\sc Feedback
  Vertex Set}, for which the associated hereditary graph property is
the collection of all  forests; the forbidden set is
the set of all cycles and hence is not finite.  Indeed,
Theorem~\ref{fullkernelresult} does not apply to {\sc Feedback Vertex
  Set} either, since the number of minimal solutions exceeds $f(k)$ if
the input graph includes a cycle of length $f(k)+1$, for any function
$f$.  While it maybe possible to adapt the compact enumeration of
minimal feedback vertex sets~\cite{Guoetal2006} for reconfiguration,
we develop a reconfiguration kernel for feedback vertex set by
modifying a specific kernel for the problem.

We are given an undirected graph and two feedback vertex sets $S$ and $T$ of size at most $k$.
We make use of Bodlaender's cubic kernel for {\sc Feedback Vertex
  Set}~\cite{B07}, modifying reduction rules (shown in italics) to
allow the reconfiguration sequence to use non-minimal solutions, and
to take into account the roles of $C$, $S_D$, $T_A$, and $O$.  In some
cases we remove vertices from $O$ only, as others may be needed in a
reconfiguration sequence.

The reduction may introduce multiple edges, forming a
multigraph. Bodlaender specifies that a double edge between vertices
$u$ and $v$ consists of two edges with $u$ and $v$ as endpoints.
Since we preserve certain degree-two vertices, we extend the notion by
saying that there is a {\em double edge} between $u$ and $v$ if either
there are two edges with $u$ and $v$ as endpoints, one edge between
$u$ and $v$ and one path from $u$ to $v$ in which each internal vertex
is of degree two, or two paths (necessarily sharing only $u$ and $v$)
from $u$ to $v$ in which each internal vertex is of degree two. Following
Bodlaender, we define two sets of vertices, a feedback vertex set $A$
of size at most $2k$ and the set $B$ containing each vertex with a
double edge to at least one vertex in $A$.  A {\em piece} is a
connected component of $G[V \setminus (A \cup B)]$, the {\em border}
of a piece with vertex set $X$ is the set of vertices in $A \cup B$
adjacent to any vertex in $X$, and a vertex $v$ in the border {\em
  governs} a piece if there is a double edge between $v$ and each
other vertex in the border.  We introduce ${\cal E}$ to denote how
much capacity we can ``free up'' for use in the reduced instance by
removing vertices and then readding them.

Bodlaender's algorithm makes use of a repeated initialization phase in
which an approximate solution $A$ is found and $B$ is initialized; for
our purposes, we set $A = C \cup S_D \cup T_A$ in the first round and
thereafter remove vertices as dictated by the application of reduction
rules.  Although not strictly necessary, we preserve this idea in
order to be able to apply Bodlaender's counting arguments.  In the following rules, $v$, $w$, and $x$ are vertices.

\begin{description}
\item[Rule 1] If $v$ has degree 0, remove $v$ from $G$. {\em If $v$ is in $S_D \cup T_A$, subtract 1 from $\ell$.  If $v$ is in $C$, increment ${\cal E}$ by 1.} \item[Rule 2] If $v$ has degree 1, remove $v$ and its incident edge from $G$.
{\em If $v$ is in $S_D \cup T_A$, subtract 1 from $\ell$.  If $v$ is in $C$, increment ${\cal E}$ by 1.}
\item[Rule 3] If there are three or more edges $\{v,w\}$, remove all but two.
\item[Rule 4] If $v$ has degree 2 {\em and $v$ is in $O$}, remove $v$
  and its incident edges from $G$ and add an edge between its
  neighbours $w$ and $x$; add $w$ (respectively, $x$) to $B$ if a
  double edge is formed, $w$ (respectively, $x$) is not in $A \cup B$,
  and $x$ (respectively, $w$) is in $A$.
\item[Rule 5] If $v$ has a self-loop, remove $v$ and all incident edges and decrease $k$ by 1, then restart the initialization phase.
\item[Rule 6]  If there are at least $k+2$ vertex-disjoint paths between $v \in A$ and any $w$ and there is no double edge between $v$ and $w$, add two edges between $v$ and $w$, and if $w \notin A \cup B$, add $w$ to $B$.
\item[Rule 7] If for $v \in A$ there exist at least $k+1$ cycles such that each pair of cycles has exactly $\{v\}$ as the intersection, remove $v$ and all incident edges and decrease $k$ by 1, then restart the initialization phase.
\item[Rule 8] If $v$ has at least $k+1$ neighbours with double edges, remove $v$ and all incident edges and decrease $k$ by 1, then restart the initialization phase.
\item[Rule 9] If $v \in A \cup B$ governs a piece with vertex set $X$ and has exactly one neighbour $w$ in $X$, then remove the edge $\{v,w\}$.
\item[Rule 10] If $v \in A \cup B$ governs a piece with vertex set $X$ and has at least two neighbours in $X$, then remove $v$ and all incident edges and decrease $k$ by 1, then restart the initialization phase. {\em Replaced by the following rule: If a piece with vertex set $X$ has a border set $Y$ such that there is a double edge between each pair of vertices in $Y$, remove $X$.}
\end{description}

\begin{lemma}\label{lemma-fvs-reconf}
The instance $(G,S,T,k,\ell)$ is a yes-instance if and only if one of
the ${\cal E} +1$ reduced instances $(G', S', T', k-e, \ell-2({\cal E}
- e))$, for $0 \le e \le {\cal E}$, is a yes-instance.
\end{lemma}

\begin{proof}
We show that no modification of a reduction rule removes possible
reconfiguration sequences.  This is trivially true for Rules 3 and 6.

The vertices removed by Rules 1, 2, and 4 play different roles in
converting a reconfiguration sequence for a reduced instance to a
reconfiguration sequence for the original instance.  As there is no
cycle that can be destroyed only by a vertex removed from $O$ by Rule
1, 2, or 4, none of these vertices are needed.  To account for the
required removal (addition) of each such vertex in $S_D$
($T_A$), we remove all $d$ such vertices and decrease $\ell$ by $d$.
We can choose
to leave a $v \in C_M$ in each solution in the sequence (with
no impact on $\ell$) or to remove and then readd $v$ to free up extra capacity,
at a cost of incrementing $\ell$ by two; in the reduced instance we thus remove $v$ and either decrement $k$ or subtract two from $\ell$.  Since this choice can be made for each
of these vertices, ${\cal E}$ in total, we try to solve any of ${\cal
  E} +1$ versions $(G', S', T', k-e, \ell-2({\cal E} - e))$ for $0 \le
e \le {\cal E}$.

For each of Rules 5, 7, and 8, we show that the removed vertex $v$ is in
$C_F$; since the cycles formed by $v$ must be handled by each solution
in the sequence, the instance can be reduced by
removing $v$ and decrementing $k$. For Rule 5, $v \in C_F$ since every
feedback arc set must contain $v$.  For Rules 7 and 8,
$v \in C_F$, since any feedback vertex set not
containing $v$ would have to contain at least $k+1$ vertices, one for
each cycle.

For Rule 9, Bodlaender's Lemma 8 shows that the removed edge has no impact on
feedback vertex sets.

For Rule 10, we first assume that Rule 9 has been exhaustively
applied, and thus each vertex in the border has two edges to $X$.  By
Fact~\ref{fact-pieces} for $\pi$ the set of acyclic graphs,
there cannot be a cycle in $G[O \cup \{v\}]$ for any $v \in S_D \cup
T_A \cup O$, and hence each member of the border is in $C$.  Lemma 9
in Bodlaender's paper shows that there is a minimum size feedback
vertex set containing $v$: even if all the neighbours of $v$ in the
border are included in a feedback vertex set, at least one more vertex
is required to break the cycle formed by $v$ and $X$.
There is no gain in capacity possible by replacing $v$ in the
reconfiguration sequence, and hence this particular piece is of no
value in finding a solution.\qed
\end{proof}

%
%

We first present the key points and lemmas in Bodlaender's
counting argument and then show that, with minor modifications, the
same argument goes through for our modified reduction rules and
altered definition of {\em double edge}.

In Bodlaender's proof, the size of the reduced instance is bounded by
bounding the sizes of $A$ and $B$ (Lemma~\ref{lemma-bod-AB}), bounding
the number of pieces (Lemma~\ref{lemma-bod-piece-count}), and bounding
the size of each piece. Crucial to the proof of
Lemma~\ref{lemma-bod-piece-count} is Lemma~\ref{lemma-bod-not-double},
as the counting  associates each piece with a pair of
vertices in its border that are not connected by a double edge and
then counts the number of pieces associated with each different type of
pair.  We use Lemma~\ref{lemma-bod-9} in the discussion below.

\begin{lemma}\cite{B07}\label{lemma-bod-9}
Suppose $v \in A \cup B$ governs a piece with vertex set $X$.  Suppose
there are at least two edges with one endpoints $v$ and one endpoint
in $X$.  Then there is a minimum size feedback vertex set in $G$ that
contains $v$.
\end{lemma}

\begin{lemma}\cite{B07}\label{lemma-bod-AB}
In a reduced instance, there are at most $2k$ vertices in $A$ and at
most $2k^2$ vertices in $B$.
\end{lemma}

\begin{lemma}\cite{B07}\label{lemma-bod-not-double}
Suppose none of the Rules 1--10 can be applied to $G$.  Suppose $Y
\subseteq V$ is the border of a piece in $G$.  Then there are two
disjoint vertices $v,w \in Y$ such that $\{v,w\}$ is not a double
edge.
\end{lemma}

\begin{lemma}\cite{B07}\label{lemma-bod-piece-count}
Suppose we have a reduced instance.  There are at most $8k^3 + 9k^2 + k$ pieces.
\end{lemma}

\begin{lemma}\label{lemma-fvs-bod}
Each reduced instance has $O(k^3)$ vertices and $O(k^3)$ edges, and can
be obtained in polynomial time.
\end{lemma}

\begin{proof}

  Our modifications to Rules 1--3 and 5--9 do not have an impact on
  the size of the kernel.  Although our Rule 4 preserves some vertices
  in $A$ of degree two, due to the initialization of $A$ to be $C \cup
  S_D \cup T_A$, and hence of size at most $2k$, the bound on $B$ and
  hence Lemma~\ref{lemma-bod-AB} follows from Rule 8.  In essence, our
  extended definition of double edges handles the degree-two vertices
  that in Bodlaender's constructions would have been replaced by an
  edge.

  To claim the result of Lemma~\ref{lemma-bod-piece-count}, it
  suffices to show that Lemma~\ref{lemma-bod-not-double} holds for our
  modified rules.  Bodlaender shows that if there is a piece such that
  each pair of vertices in the border set is connected by a double
  edge, Rule 10 along with Rule 9 can be applied repeatedly to remove
  vertices from the border of the piece and thereafter Rules 2 and 1
  to remove the piece entirely.

  To justify Rule 10, Bodlaender shows in Lemma~\ref{lemma-bod-9} that
  if $v \in A \cup B$ governs a piece with vertex set $X$ and there
  are at least two edges between $v$ and $X$, then there is a minimum
  size feedback vertex set in $G$ that contains $v$.  For our
  purposes, however, since there may be non-minimum size feedback
  vertex sets used in the reconfiguration sequence, we wish to retain
  $v$ rather than removing it.
  Our modification to Rule 10 allows us to retain $v$, handling all
  the removals from the piece without changing the border, and thus
  establishing Lemma~\ref{lemma-bod-not-double}, as needed to prove
  Lemma~\ref{lemma-bod-piece-count}.

  In counting the sizes of pieces, our modifications result in extra
  degree-two vertices.  Rule 4 removes all degree-two vertices in $O$, and
  hence the number of extra vertices is at most $2k$, having no effect
  on the asymptotic count.  \qed
\end{proof}

\begin{theorem}\label{theorem-fpt-fvs}
  \textsc{Feedback Vertex Set Reconfiguration} and the search variant parameterized by $k$
are in {\em FPT}.
\end{theorem}

\begin{proof}
  Since the number of reduced instances
  is ${\cal E} + 1 \le |C| + 1 \le k + 1$,
  as a consequence of Lemmas~\ref{lemma-fvs-reconf} and
  \ref{lemma-fvs-bod}, we have a reconfiguration kernel, proving the
  first result.

  For the search version, we observe that we can generate the
  reconfiguration graph of the reduced yes-instance and use it to
  extract a reconfiguration sequence.
  We demonstrate that we can form a reconfiguration sequence for
  $(G,S,T,k,\ell)$ from the reconfiguration sequence $\sigma$ for the
  reduced yes-instance $(G',S',t',k-e,\ell-2({\cal E}-e))$.  We choose an
  arbitrary partition of the vertices removed from $G$ by Rules 1 and
  2 into two sets, $K$ (the ones to keep) of size $e$ and $M$ (the
  ones to modify) of size ${\cal E}-e$.  We can modify $\sigma$ into a
  sequence $\sigma'$ in which all vertices in $K$ are added to each
  set; clearly no set will have size greater than $k$.  Our
  reconfiguration sequence then consists of ${\cal E}-e$ steps each
  deleting an element of $M$, the sequence $\sigma'$, and ${\cal E}-e$ steps each
  adding an element of $M$, for a length of at most $({\cal E} -e) +
  (\ell -({\cal E}-e)) + ({\cal E} - e) \le \ell$, as needed. \qed
\end{proof}

\section{Hardness Results}\label{sec-relate}

The reductions presented in this section make use of the forbidden set
characterization of heredity properties.  A {\em $\pi$-critical graph}
$H$ is a (minimal) graph in the forbidden set $\Fh_\pi$ that has at
least two vertices; we use the fact that $H \notin \pi$, but the
deletion of any vertex from $H$ results in a graph in $\pi$.  For
convenience, we will refer to two of the vertices in a $\pi$-critical
graph as {\em terminals} and the rest as {\em internal vertices}. We
construct graphs from multiple copies of $H$.  For a positive integer
$c$, we let $H_c^*$ be the (``star'') graph obtained from each of $c$ copies $H_i$ of $H$
by identifying an arbitrary terminal $v_i$, $1 \le i \le c$, from each
$H_i$;
in
$H_c^*$
vertices $v_1$ through $v_c$ are replaced with a vertex $w$, the {\em
  gluing vertex of $v_1$ to $v_c$}, to form a graph with vertex set
$\cup_{1 \le i \le c} (V(H_i) \setminus \{v_i\}) \cup \{w\}$ and edge
set $\cup_{1 \le i \le c}\{\{u,v\} \in E(H_i) \mid v_i \notin
\{u,v\}\} \cup \cup_{1 \le i \le c}\{\{u,w\} \mid \{u,v_i\} \in E(H_i)\}$.
A terminal is {\em non-identified} if it is not used in forming a
gluing vertex.

In Figure~\ref{fig-star}, $H$ is a $K_3$ with terminals marked black
and gray; $H_4^*$ is formed by identifying all the gray terminals to
form $w$.

\begin{figure}
\begin{centering}
\centerline{\includegraphics[scale=0.45]{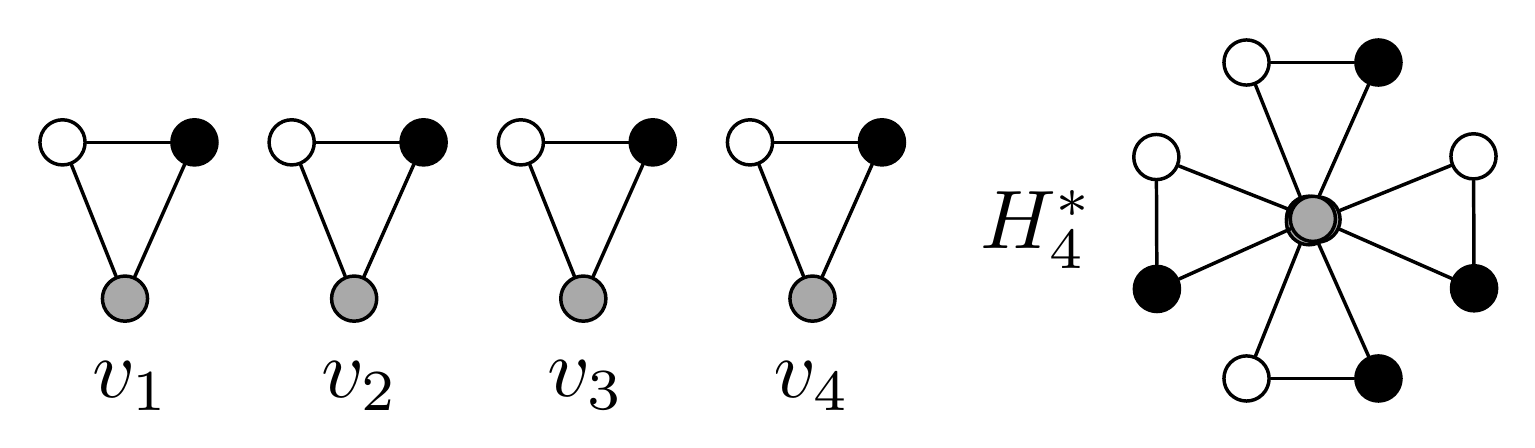}}
\end{centering}
\caption{An example $H_c^*$}
\label{fig-star}
\end{figure}

\begin{theorem}\label{theorem-reconf-minmax}
Let $\pi$ be any hereditary property satisfying the following:
\begin{itemize}
\item
For any two graphs $G_1$ and $G_2$ in $\pi$, the graph obtained by their disjoint union is in $\pi$.
\item
There exists an $H \in \Fh_\pi$ such that if $H_c^*$ is the graph obtained
from identifying a terminal from each of $c$ copies of $H$, then the graph
$R= H_c^*[V(H_c^*) \setminus \{ u_1, u_2, \ldots u_c \}]$ is in $\pi$, where $u_1, u_2, \ldots u_c$
are the non-identified terminals in the $c$ copies of $H$.
\end{itemize}
Then each of the following is at least as hard as \textsc{$\pi$-subset$(G, k)$}:
\begin{enumerate}
\item
\textsc{$\pi$-del-reconf$(G, S, T, k, \ell)$} parameterized by $\ell$, and
\item
\textsc{$\pi$-sub-reconf$(G, S, T, k, \ell)$} parameterized by $k + \ell$.
\end{enumerate}
\end{theorem}

\begin{proof}
  Given an instance of \textsc{$\pi$-subset$(G, k)$} and a $\pi$-critical
  graph $H$ satisfying the hypothesis of the lemma, we form an
  instance of \textsc{$\pi$-del-reconf$(G', S, T, |V(G)| + k, 4k)$}, with
  $G'$, $S$, and $T$ defined below.  The graph $G'$ is the disjoint
  union of $G$ and a graph $W$ formed from $k^2$ copies of $H$, where
  $H_{i,j}$ has terminals $\ell_{i,j}$ and $r_{i,j}$.  We let $a_i$, $1 \le i \le
  k$, be the gluing vertex of $\ell_{i, 1}$ through
  $\ell_{i, k}$, and let $b_j$, $1 \le j \le k$, be the gluing
  vertex of $r_{1, j}$ through $r_{k, j}$, so that there is a
  copy of $H$ joining each  $a_i$ and $b_j$. An example $W$ is shown
  in Figure~\ref{fig-thm17}, where copies of $H$ are shown
  schematically as gray ovals.  We let $A = \{a_i
  \mid 1 \le i \le k\}$, $B = \{b_j \mid 1 \le j \le k \}$, $S = V(G)
  \cup A$, and $T = V(G) \cup B$.  Clearly $|V(G')| = |V(G)| + 2k +
  k^2 (|V(H)| -2)$ and $|S|=|T|= |V(G)| + k$. Moreover, each of $V(G')
  \setminus S$ and $V(G') \setminus T$ induce a graph in $\pi$, as
  each consists of $k$ disjoint copies of $H_k^*$ with one of the
  terminals removed from each $H$ in $H_k^*$.

\begin{figure}
\begin{centering}
\centerline{\includegraphics[scale=0.35]{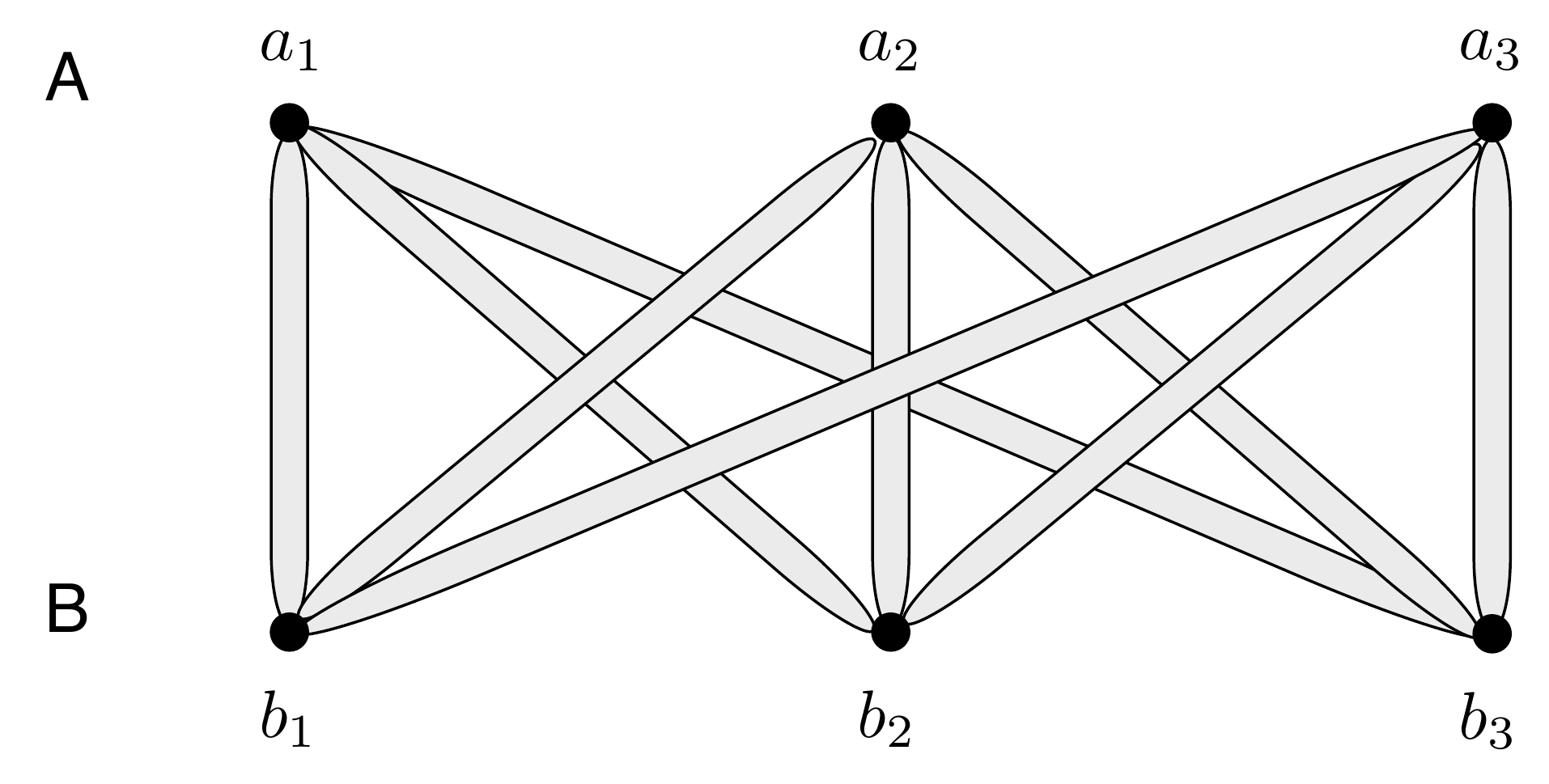}}
\end{centering}
\caption{An example $W$}
\label{fig-thm17}
\end{figure}

  Suppose the instance of \textsc{$\pi$-del-reconf$(G', S, T, |V(G)|+k, 4k)$}
  is a yes-instance.  As there is a copy of $H$ joining each
  vertex of $A$ to each vertex of $B$, before deleting $a \in A$ from
  $S$ the reconfiguration sequence must add all of $B$ to ensure that
  the complement of each intermediate set induces a graph in $\pi$.
  Otherwise, the complement will contain at least one copy of $H$ as a subgraph
  and is therefore not in $\pi$.
  The capacity bound of $|V(G)| + k$ implies that the reconfiguration
  sequence must have
  deleted from $S$ a subset $S' \subseteq V(G)$ of size at least $k$
  such that $V(G') \setminus (S \setminus S') = S' \cup B$ induces a
  subgraph in $\pi$. Thus, $G[S'] \in \pi$, and hence
  \textsc{$\pi$-subset$(G, k)$} is a yes-instance.

  Conversely if the instance of \textsc{$\pi$-subset$(G, k)$} is a
  yes-instance, then there exists $V' \subseteq V(G)$ such that $|V'|
  = k$ and $G[V'] \in \pi$.  We form a reconfiguration sequence
  between $S$ and $T$ by first deleting all vertices in $V'$ from $S$
  to yield a set of size $|V(G)|$.  $G'[V(G') \setminus (S \setminus
  V')]$ consists of the union of $G'[V'(G) \setminus S]$ and $G'[V'] =
  G[V']$, both of which are in $\pi$.  Next we add one by one all
  vertices of $B$, then delete one by one all vertices of $A$ and then
  add back one by one each vertex in the set $V'$ resulting in a
  reconfiguration sequence of length $k + k + k + k = 4k$. It is clear
  that in every step, the complement of the set induces a graph in
  $\pi$.

  Thus we have showed that \textsc{$\pi$-subset$(G, k)$} is a yes-instance
  if and only if there is a path of length at most $4k$ between $S$
  and $T$ in $R^{\pi}_{\textsc{del}}(G', |V(G)| + k)$.
  Since $|V(G')| - (|V(G)| + k) = k + k^2 (|V(H)| -2))$, this implies that
  \textsc{$\pi$-subset$(G, k)$} is a yes-instance if and only if there is
  a path of length at most $4k$ between ${V(G') \setminus S}$ and
  ${V(G') \setminus T}$ in $R^{\pi}_{\textsc{sub}}(G', k + k^2 (|V(H)| -2))$.
  Therefore, \textsc{$\pi$-sub-reconf$(G, S, T, k, \ell)$}
  parameterized by $k + \ell$ is at least as hard as
  \textsc{$\pi$-subset$(G, k)$}, proving the second part.\qed
\end{proof}

\begin{corollary}\label{corollary:hardness}
\textsc{Vertex Cover Reconfiguration}, \textsc{Feedback Vertex Set Reconfiguration},
and \textsc{Odd Cycle Transversal Reconfiguration} parameterized by $\ell$
are all $W[1]$-hard and
\textsc{Independent Set Reconfiguration}, \textsc{Forest Reconfiguration},
and \textsc{Bipartite Subgraph Reconfiguration} parameterized by $k + \ell$ are
all $W[1]$-hard.
\end{corollary}

\begin{proof}
It is known that for any hereditary property $\pi$ that consists of
all edgeless graphs but not all cliques~\cite{KR02}, \textsc{$\pi$-subset$(G,k)$} is $W[1]$-hard.
It is clear that the collections
of all edgeless graphs, of all bipartite graphs, and of all forests
satisfy this condition for hardness, as well as the hypothesis of
Theorem~\ref{theorem-reconf-minmax}.

For the collection of independent sets, the only $H \in \Fh_\pi$ is an
edge both of whose endpoints are terminals. Here identifying multiple copies of
$H$ at a terminal forms a star, and deleting the non-identified
terminal from each of the edges results in a single vertex, which is in
$\pi$.

For the collection of forests, and bipartite graphs, we let $H \in
\Fh_\pi$ be a triangle. When we identify multiple
triangles at a vertex, and remove another vertex of each of the
triangles, we obtain a tree, which is in $\pi$. \qed
\end{proof}

We obtain further results for properties not covered by
Theorem~\ref{theorem-reconf-minmax}.
Lemma~\ref{lemma-clique-cluster-hard} handles the collection of all
cliques, which does not satisfy the first condition of the theorem and
the collection of all {\em cluster graphs} (disjoint unions of
cliques), which satisfies the first condition but not the second.
Moreover, as \textsc{$\pi$-subset$(G, k)$} is
in {\em FPT} for $\pi$ the collection of all cluster
graphs~\cite{KR02}, Theorem~\ref{theorem-reconf-minmax} provides no
lower bounds.

\begin{lemma}\label{lemma-clique-cluster-hard}
\textsc{Clique Reconfiguration} and \textsc{Cluster Subgraph Reconfiguration}
parameterized by $k + \ell$ are $W[1]$-hard.
\end{lemma}

\begin{proof}
We first give an {\em FPT} reduction from \textsc{$t$-Clique}, known
to be $W[1]$-hard, to \textsc{Cluster Subgraph Reconfiguration}.  For
$(G, t)$ an instance of \textsc{$t$-Clique}, $V(G) = \{v_1, \ldots,
v_n\}$, we form a graph consisting of four $K_t$'s (with vertex sets
$A$, $B$, $C$, and $D$) and a subgraph mimicking $G$ (with vertex set
$X$), where there is an edge from each vertex in $X$ to each vertex in
each $K_t$, and each of subgraphs induced on the following vertex sets
induce a $K_{2t}$: $A \cup B$, $A \cup C$, $B \cup D$, $C \cup D$.
More formally, $G' = (X \cup A \cup B \cup C \cup D, E_X \cup E_T \cup E_C)$, where $X = \{x_1, \ldots, x_n\}$, $|A| = |B| = |C| = |D| = t$, $E_X = \{\{x_i,x_j\} \mid \{v_i,v_j\} \in E(G)\}$ corresponds to the edges in $G$, $E_T = \{\{a,a'\} \mid a,a' \in A, a \ne a'\} \cup
\{\{b,b'\} \mid b,b' \in B, b \ne b'\} \cup
\{\{c,c'\} \mid c,c' \in C, c \ne c'\} \cup
\{\{d,d'\} \mid d,d' \in D, d \ne d'\} $ forms the $K_t$ cliques, and
$E_C = \{\{x,a\}, \{x,b\}, \{x,c\}, \{x,d\}, \{a,b\}, \{a,c\}, \{b,d\}, \{c,d\} \mid a \in A, b \in B, c \in C, d \in D, x \in X\}$ forms the connections among the vertex setes.

We let $(G', S, T, 2t, 6t)$ be an instance of \textsc{Cluster Subgraph
  Reconfiguration}, where $S = A \cup B$ and $T = C \cup D$.  Clearly
$|S| = |T| = 2t$ and both $S$ and $T$ induce cluster graphs (in fact
cliques). We claim that $G$ has a clique of size $t$ if and only if
there is a path of length $6t$ from $S$ to $T$.

If $G$ has a clique of size $t$, then there exists a subset $Y
\subseteq X$ forming a clique of size $t$.  We form a reconfiguration
sequence of length $6t$ as follows; add the vertices $Y$, remove the
vertices in $A$, add the vertices in $D$, remove the vertices in $B$,
add the vertices in $C$, and remove the vertices in $Y$, one by one.
It is not hard to see that at every step in this sequence we maintain
an induced clique in $G'$ of size greater than or equal to $2t$ (and
hence a cluster subgraph).

If there exists a path of length $6t$ from $S$ to $T$, we make use of
the fact that no cluster subgraph contains an induced path of length
three to show that $G$ has a clique of size $t$.  Observe that before
adding any vertex of $C$, we first need to remove (at least) all of
$B$ since otherwise we obtain an induced path of length three
containing vertices in $C$, $A$, and $B$, respectively.  Similarly, we
cannot add any vertex of $D$ until we have removed all of $A$.
Therefore, before adding any vertex from $T$, we first need to delete
at least $t$ vertices from $S$.  To do so without violating our
minimum capacity of $2t$, at least $t$ vertices must be added from
$X$.  Since every vertex in $X$ is connected to all vertices in $S$
and $T$, if any pair of those $t$ vertices do not share an edge, we
obtain an induced path on three vertices. Thus $X$, and hence $G$,
must have a clique of size $t$.

Since in our reduction $S$ and $T$ are cliques and every reconfiguration
step maintains an induced clique in $G'$ of size greater than or equal to $2t$,
the same applies to the \textsc{Clique Reconfiguration} problem.
Consequently, both \textsc{Clique Reconfiguration} and \textsc{Cluster Subgraph Reconfiguration} parameterized by $k + \ell$ are $W[1]$-hard.
\qed
\end{proof}

As neither \textsc{Dominating Set} nor its parametric dual is a
hereditary graph property, Theorem~\ref{theorem-reconf-minmax} is
inapplicable; we instead use a construction specific to this
problem in Lemma~\ref{lemma-dom-hard}, which in turn leads to
Corollary~\ref{corollary-hitting-set}, since \textsc{Dominating Set}
can be phrased as a hitting set of the family of closed neighborhood
of the vertices of the graph.

\begin{lemma}\label{lemma-dom-hard}
\textsc{Dominating Set Reconfiguration} parameterized
by $k + \ell$ is $W[2]$-hard.
\end{lemma}

\begin{proof}
We give a reduction from \textsc{$t$-Dominating Set};
for $(G, t)$ an instance of \textsc{$t$-Dominating Set}, we
form $G'$ as the disjoint union of two graphs $G'_1$ and $G'_2$.

We form $G'_1$ from  $t+2$ $(t+1)$-cliques
$C_0$ (the {\em outer clique}) and $C_1$, \ldots, $C_{t+1}$ (the
{\em inner cliques}); $V(C_0) = \{o_1, \ldots, o_{t+1}\}$ and
$V(C_i) = \{w_{(i,0)},w_{(i,1)},\ldots w_{(i,t)}\}$ for $1 \le i
\le t+1$.  The edge set of $G'_1$ contains not only the edges of
the cliques but also $\{\{o_j,w_{(i,j)}\} \mid 1 \le i \le t+1, 0
\le j \le t\}$; the graph to the left in Figure~\ref{fig-dom}
illustrates $G'_1$ for $t= 2$.  Any dominating set that
does not contain all vertices in the outer clique must contain a
vertex from each inner clique.

\begin{figure}
\begin{centering}
\centerline{\includegraphics[scale=0.45]{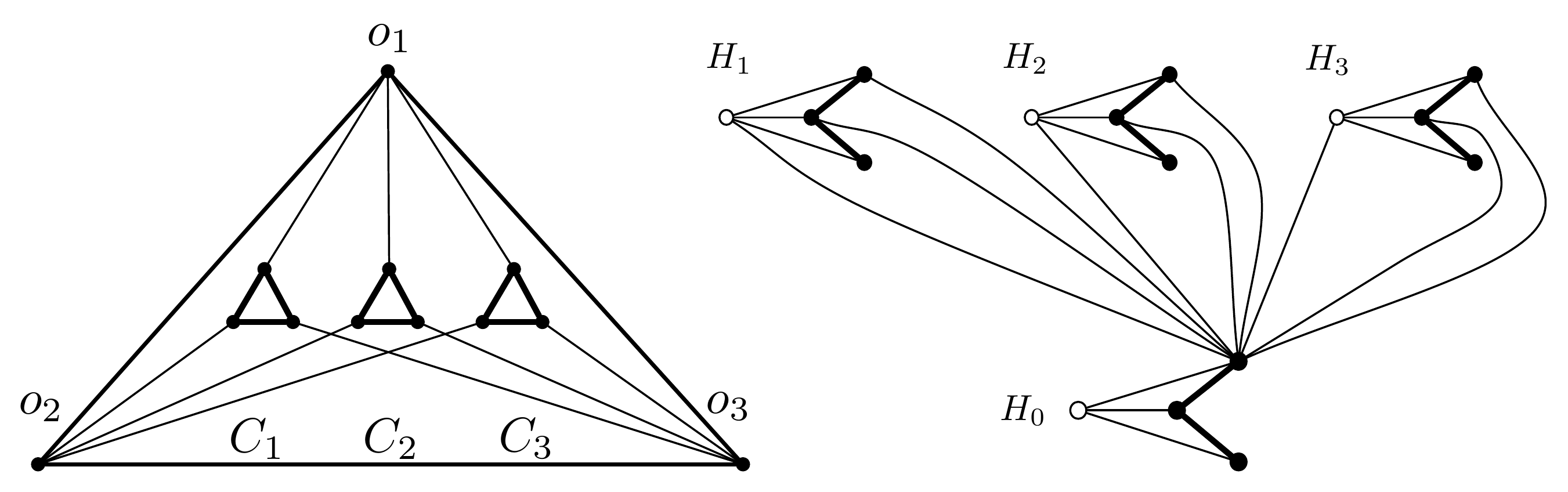}}
\end{centering}
\caption{Graphs used for the dominating set reduction}
\label{fig-dom}
\end{figure}

To create $G'_2$, we first define $G^+$ to be the graph formed by
adding a universal vertex to $G$, where we assume without loss of
generality that $V(G) = \{v_1,\ldots,v_{|V(G)|}\}$.  We let $V(G'_2) =
\cup_{0 \le i \le t}V(H_i)$, where $H_0, \ldots, H_{t}$ are
$t+1$ copies of $G^+$; we use $u_i$ to denote the universal vertex
in $H_i$ and $v_{(i,j)}$ to denote the copy of $v_j$ in $H_i$, $1 \le
j \le |V(G)|$, $0 \le i \le t$.  The edge set consists of edges
between each non-universal vertex $v_{(0,j)}$ in $H_0$ and, in each
$H_i$, the universal vertex, its image, and the images of its
neighbours in $G$, or more formally $E(G'_2) = \{\{v_{0,j},u_i\} \mid
1 \le j \le |V(G)|, 1 \le i \le t\} \cup \{\{v_{0,j},v_{i,j}\} \mid
1 \le j \le |V(G)|, 1 \le i \le t\} \cup \{\{(v_{0,j},v_{i,k}\}
\mid 1 \le j \le |V(G)|, 1 \le i \le t, (v_j,v_k) \in E(G)\}$.  The
graph to the right in Figure~\ref{fig-dom} illustrates part of $G'_2$,
where universal vertices are shown in white and, for the sake of
readability, the only edges outside of $G^+$ shown are those adjacent
to a single vertex in $H_0$.

We form an instance $(G', S, T, 3t + 2, 6t + 4)$ of
\textsc{Dominating Set Reconfiguration}, where $S = \{u_i \mid 0 \le i
\le t\} \cup V(C_0)$ and $T = \{u_i \mid 0 \le i \le t\} \cup
\{w_{i,i-1} \mid 1 \le i \le t+1\}$.  Both $S$ and $T$ are
dominating sets, as each universal vertex $u_i$
dominates $H_i$ as well as $H_0$ and $V(G'_1)$ is dominated by
the outer clique in $S$ and by one vertex from each inner clique in
$T$. Clearly $|S| = |T| = 2t+ 2$.

We claim that $G$ has a dominating set of size $t$ if and only if
there is a path of length $6t + 4$ from $S$ to $T$. In $G'_1$, to remove any vertex from the outer clique, we
must first add a vertex from each inner clique, for a total of $t + 1$ additions;
since $k = 3t + 2$ and $|S| = 2t + 2$, this
can only take place after $G'_2$ has been dominated using at most
$t$ vertices.
In $G'_2$, a universal vertex $u_i$ cannot be deleted
until $H_i$ has been dominated. If $G$ can be dominated with $t$
vertices, then it is possible to add the dominating set in $H_0$
and remove all the universal vertices, thus making the required capacity available.
If not, then none of the universal vertices, say $u_i$, can be
removed without first adding at least $t+1$ vertices to dominate
$H_i$, for which there is not enough capacity. Therefore, there exists a
reconfiguration sequence from $S$ to some $S'$ such that $S' \cap G'_2$
has $t$ vertices if and only if $G$ has a dominating
set of size $t$. Moreover, the existence of a dominating set $D$
of size $t$ in $G$ implies a path of length $6t + 4$ from $S$
to $T$; we add $D$ in $H_0$, remove all universal vertices, reconfigure
$G'_1$, add all universal vertices, and then remove $D$. Consequently,
there exists a reconfiguration sequence from $S$ to $T$ in $6t + 4$
steps if and only if $G$ has a dominating set of size $t$.
\qed
\end{proof}

The following is a result of there being a polynomial-time
parameter-preserving reduction from \textsc{Dominating Set}:

\begin{corollary}\label{corollary-hitting-set}
\textsc{Unbounded Hitting Set Reconfiguration} parameterized by $k + \ell$
is $W[2]$-hard.
\end{corollary}

\vspace{-0.2cm}
\section{Conclusions and Directions for Further Work}
Our results constitute the first study of the parameterized complexity
of reconfiguration problems.  We give a general paradigm, the reconfiguration kernel, for
proving fixed-parameter tractability,
and provide hardness reductions that apply to problems associated with
hereditary graph properties.
Our result on cluster graphs (Lemma~\ref{lemma-clique-cluster-hard})
demonstrates the existence of a problem that is fixed-parameter
tractable~\cite{KR02}, but whose reconfiguration version is $W$-hard
when parameterized by $k$; this clearly implies that fixed-parameter
tractability of the underlying problem does not guarantee
fixed-parameter tractability of reconfiguration when parameterized by $k$.  
Since there is
unlikely to be a polynomial-sized kernel for the problem of
determining whether a given graph has a cluster of size at least
$k$~\cite{KPRR12}, it is possible (though in our opinion, unlikely)
that an underlying problem having a polynomial-sized kernel is
sufficient for the reconfiguration problem to be fixed-parameter
tractable when parameterized by $k$.

 

It remains open whether there exists an NP-hard problem
for which the reconfiguration version is in {\em FPT} when parameterized by
$\ell$.


Our {\em FPT} algorithms for reconfiguration of \textsc{Bounded Hitting Set} and \textsc{Feedback Vertex Set} have running times of $O^*(2^{O(k \lg k)})$.  Further work is needed to determine whether the running times can be improved to
$O^*(2^{O(k)})$, or whether these bounds are tight under the {\em Exponential Time Hypothesis}.

We observe connections to another well-studied paradigm, local
search~\cite{FRFLSV90}, where the aim is to find
an {\em improved solution} at distance $\ell$ of a given
solution $S$.  Not surprisingly, as in
local search, the problems we study turn out to be hard even in the
parameterized setting when parameterized by $\ell$.
Other natural directions to
pursue (as in the study of local search) are the parameterized complexity of reconfiguration problems in special classes of graphs and of non-graph reconfiguration problems,
as well as other parameterizations.


\vspace{-0.1cm}
\subsection*{Acknowledgements}
The second author wishes to thank Marcin Kami\'{n}ski for suggesting
the examination of reconfiguration in the parameterized setting.

\bibliographystyle{acm}	
\bibliography{references}

\begin{thebibliography}{10}

\bibitem{Bang2008}
{\sc Bang-Jensen, J., and Gutin, G.}
\newblock {\em Digraphs: Theory, Algorithms and Applications}.
\newblock Springer monographs in mathematics. Springer, 2008.

\bibitem{B07}
{\sc Bodlaender, H.~L.}
\newblock A cubic kernel for feedback vertex set.
\newblock In {\em Proc. of the 24th Annual Conference on Theoretical Aspects of
  Computer Science\/} (2007), pp.~320--331.

\bibitem{BB13}
{\sc Bonamy, M., and Bousquet, N.}
\newblock Recoloring bounded treewidth graphs.
\newblock In {\em Proc. of Latin-American Algorithms Graphs and Optimization
  Symposium\/} (2013).

\bibitem{B12}
{\sc Bonsma, P.}
\newblock The complexity of rerouting shortest paths.
\newblock In {\em Proc. of Mathematical Foundations of Computer Science\/}
  (2012), pp.~222--233.

\bibitem{BC09}
{\sc Bonsma, P.~S., and Cereceda, L.}
\newblock Finding paths between graph colourings: {PSPACE}-completeness and
  superpolynomial distances.
\newblock {\em Theor. Comput. Sci. 410}, 50 (2009), 5215--5226.

\bibitem{CVJ08}
{\sc Cereceda, L., van~den Heuvel, J., and Johnson, M.}
\newblock Connectedness of the graph of vertex-colourings.
\newblock {\em Discrete Mathematics 308}, 56 (2008), 913--919.

\bibitem{CVJ09}
{\sc Cereceda, L., van~den Heuvel, J., and Johnson, M.}
\newblock Mixing 3-colourings in bipartite graphs.
\newblock {\em European Journal of Combinatorics 30}, 7 (2009), 1593--1606.

\bibitem{CVJ11}
{\sc Cereceda, L., van~den Heuvel, J., and Johnson, M.}
\newblock Finding paths between 3-colorings.
\newblock {\em Journal of Graph Theory 67}, 1 (2011), 69--82.

\bibitem{D09}
{\sc Damaschke, P., and Molokov, L.}
\newblock The union of minimal hitting sets: Parameterized combinatorial bounds
  and counting.
\newblock {\em Journal of Discrete Algorithms 7}, 4 (2009), 391--401.

\bibitem{DF97}
{\sc Downey, R.~G., and Fellows, M.~R.}
\newblock {\em Parameterized complexity}.
\newblock Springer-Verlag, New York, 1997.

\bibitem{FRFLSV90}
{\sc Fellows, M.~R., Rosamond, F.~A., Fomin, F.~V., Lokshtanov, D., Saurabh,
  S., and Villanger, Y.}
\newblock Local search: is brute-force avoidable?
\newblock In {\em Proc. of the 21st International Joint Conference on Artifical
  Intelligence\/} (2009), pp.~486--491.

\bibitem{FG06}
{\sc Flum, J., and Grohe, M.}
\newblock {\em Parameterized complexity theory}.
\newblock Springer-Verlag, Berlin, 2006.

\bibitem{GKMP09}
{\sc Gopalan, P., Kolaitis, P.~G., Maneva, E.~N., and Papadimitriou, C.~H.}
\newblock The connectivity of boolean satisfiability: computational and
  structural dichotomies.
\newblock {\em SIAM J. Comput. 38}, 6 (2009), 2330--2355.

\bibitem{Guoetal2006}
{\sc Guo, J., Gramm, J., H\"uffner, F., Niedermeier, R., and Wernicke, S.}
\newblock Compression-based fixed-parameter algorithms for feedback vertex set
  and edge bipartization.
\newblock {\em Journal of Computer and System Sciences 72}, 8 (2006),
  1386--1396.

\bibitem{HS12}
{\sc Haas, R., and Seyffarth, K.}
\newblock The $k$-{D}ominating {G}raph, 2012.
\newblock arXiv:1209.5138.

\bibitem{HD05}
{\sc Hearn, R.~A., and Demaine, E.~D.}
\newblock {PSPACE}-completeness of sliding-block puzzles and other problems
  through the nondeterministic constraint logic model of computation.
\newblock {\em Theor. Comput. Sci. 343}, 1-2 (2005), 72--96.

\bibitem{IDHPSUU11}
{\sc Ito, T., Demaine, E.~D., Harvey, N. J.~A., Papadimitriou, C.~H., Sideri,
  M., Uehara, R., and Uno, Y.}
\newblock On the complexity of reconfiguration problems.
\newblock {\em Theor. Comput. Sci. 412}, 12-14 (2011), 1054--1065.

\bibitem{IKD12}
{\sc Ito, T., Kami\'{n}ski, M., and Demaine, E.~D.}
\newblock Reconfiguration of list edge-colorings in a graph.
\newblock {\em Discrete Applied Mathematics 160}, 15 (2012), 2199--2207.

\bibitem{KMM11}
{\sc Kami\'{n}ski, M., Medvedev, P., and Milani\v{c}, M.}
\newblock Shortest paths between shortest paths.
\newblock {\em Theor. Comput. Sci. 412}, 39 (2011), 5205--5210.

\bibitem{KR02}
{\sc Khot, S., and Raman, V.}
\newblock Parameterized complexity of finding subgraphs with hereditary
  properties.
\newblock {\em Theor. Comput. Sci. 289}, 2 (2002), 997--1008.

\bibitem{KPRR12}
{\sc Kratsch, S., Pilipczuk, M., Rai, A., and Raman, V.}
\newblock Kernel lower bounds using co-nondeterminism: finding induced
  hereditary subgraphs.
\newblock In {\em Proc. of the 13th Scandinavian Symposium and Workshops on
  Algorithm Theory\/} (2012), vol.~7357, pp.~364--375.

\bibitem{LY80}
{\sc Lewis, J.~M., and Yannakakis, M.}
\newblock The node-deletion problem for hereditary properties is {NP}-complete.
\newblock {\em Journal of Computer and System Sciences 20}, 2 (1980), 219--230.

\bibitem{Misra2013}
{\sc Misra, N., Narayanaswamy, N., Raman, V., and Shankar, B.~S.}
\newblock Solving min ones 2-{SAT} as fast as vertex cover.
\newblock In {\em Proc. of the 35th International Conference on Mathematical
  Foundations of Computer Science\/} (2010), vol.~6281, pp.~549--555.
\newblock (A revised version to appear in Theoretical Computer Science, online
  version in 10.1016/j.tcs.2013.07.019).

\bibitem{N06}
{\sc Niedermeier, R.}
\newblock {\em Invitation to fixed-parameter algorithms}.
\newblock Oxford University Press, 2006.

\bibitem{Raman2006}
{\sc Raman, V., and Saurabh, S.}
\newblock Parameterized algorithms for feedback set problems and their duals in
  tournaments.
\newblock {\em Theoretical Computer Science 351}, 3 (2006), 446--458.

\bibitem{SMN13}
{\sc Suzuki, A., Mouawad, A.~E., and Nishimura, N.}
\newblock Reconfiguration of dominating sets.
\newblock submitted.

\end{thebibliography}
\newpage

\end{document}